\documentclass[runningheads]{llncs}

\usepackage{graphicx}
\usepackage{amsmath, amssymb,amsthm}
\usepackage{tikz}
\usetikzlibrary{arrows.meta,chains,decorations.pathreplacing}

\newtheorem*{theorem*}{Theorem}

\usepackage{thm-restate}

\usepackage{natbib}  
\newcommand{\stam}[1]{}

\newcommand{\ALG}{\texttt{ALG}}

\newcommand{\E}{\mathbb{E}}

\newcommand{\assign}[2]{A_{#1,#2}}
\newcommand{\assignn}[3]{A^{#3}_{#1,#2}}
\begin{document}
\title{Prophet Inequality with Competing Agents}
%
%
\author{Tomer Ezra\inst{1} \and
Michal Feldman\inst{1,2} \and
Ron Kupfer\inst{3}}

\authorrunning{Ezra et al.}
%
\institute{Tel Aviv University, Israel\\ \and
Microsoft Research\\
\and
The Hebrew University of Jerusalem, Israel\\
\email{tomer.ezra@gmail.com}, 
\email{michal.feldman@cs.tau.ac.il}, 
\email{kupfer.ron@gmail.com}}

\maketitle              
\begin{abstract}
We introduce a model of {\em competing} agents in a prophet setting, where rewards arrive online, and decisions are made immediately and irrevocably.
The rewards are unknown from the outset, but they are drawn from a known probability distribution.
In the standard prophet setting, a single agent makes selection decisions in an attempt to maximize her expected reward.
The novelty of our model is the introduction of a competition setting, where multiple agents compete over the arriving rewards, and make online selection decisions simultaneously, as rewards arrive.
If a given reward is selected by more than a single agent, ties are broken either randomly or by a fixed ranking of the agents. 
The consideration of competition turns the prophet setting from an online decision making scenario to a multi-agent game. 

For both random and ranked tie-breaking rules, we present simple threshold strategies for the agents that give them high guarantees, independent of the strategies taken by others. 
In particular, for random tie-breaking, every agent can guarantee herself at least $\frac{1}{k+1}$ of the highest reward, 
and at least $\frac{1}{2k}$ of the optimal social welfare. 
For ranked tie-breaking, the $i$th ranked agent can guarantee herself at least a half of the $i$th highest reward.
We complement these results by matching upper bounds, even with respect to equilibrium profiles.
For ranked tie-breaking rule, we also show a correspondence between the equilibrium of the $k$-agent game and the optimal strategy of a single decision maker who can select up to $k$ rewards.


\end{abstract}
	\section{Introduction}
\label{sec:intro}

In the classical prophet inequality problem a decision maker observes a sequence of $n$ non-negative real-valued rewards $v_1, \ldots, v_n$ that are drawn from known independent distributions $F_1,\ldots ,F_n$. 
At time $t$, the decision maker observes reward $v_t$, and needs to make an immediate and irrevocable decision whether or not to accept it. If she accepts $v_t$, the game terminates with value $v_t$; otherwise, the reward $v_t$ is gone forever and the game continues to the next round.
The goal of the decision maker is to maximize the expected value of the accepted reward.

This family of problems captures many real-life scenarios, such as an employer who interviews potential workers overtime, renters looking for a potential house, a person looking for a potential partner for life, and so on. More recently, starting with the work of \citet{HajiaghayiKS07}, the prophet inequality setting has been studied within the AI community in the context of market and e-commerce scenarios, with applications to pricing schemes for social welfare and revenue maximization. For a survey on a market-based treatment of the prophet inequality problem, see the survey by~\citet{lucier2017economic}.  

An algorithm $\ALG$ has a guarantee $\alpha$ if  the expected value of   $\ALG$ is at least $\alpha$, where the expectation is taken over the coin flips of the algorithm, and the probability distribution of the input. 
\citet{KrengelS77,krengel1978semiamarts} established the existence of an algorithm that gives a tight guarantee of $\frac{1}{2}\E[\max_i v_i]$.  
Later, it has been shown that this guarantee can also be obtained by a single-threshold algorithm--- an algorithm that specifies some threshold from the outset, and accepts a reward if and only if it exceeds the threshold. 
Two such thresholds have been presented by \citet{samuel1984comparison,KleinbergW19}.  
Single-threshold algorithms are simple and easy to explain and implement. 


\paragraph{Competing Agents.}
Most attention in the literature has been given to scenarios with a single decision maker. 
Motivated by the economic aspects of the problem, where competition among multiple agents is a crucial factor, we introduce a multi-agent variant of the prophet model, in which multiple agents compete over the rewards. 

In our model, a sequence of $n$ non-negative real-valued rewards $v_1, \ldots, v_n$ arrive over time, and a set of $k$ agents make immediate and irrevocable selection decisions. 
The rewards are unknown from the outset, but every reward $v_t$ is drawn independently from a known distribution $F_t$. 
Upon the arrival of reward $v_t$, its value is revealed to all agents, and every agent decides whether or not to select it.

One issue that arises in this setting is how to resolve ties among agents. That is, who gets the reward if more than one agent selects it. We consider two natural tie-breaking rules; namely, {\em random} tie breaking (where ties are broken uniformly at random) and {\em ranked} tie-breaking (where agents are a-priori ranked by some global order, and ties are broken in favor of higher ranked agents). Random tie-breaking fits scenarios with symmetric agents, whereas ranked tie-breaking fits scenarios where some agents are preferred over others, according to some global preference order. For example, it is reasonable to assume that a higher-position/salary job is preferred over lower-position/salary job, or that firms in some industry are globally ordered from most to least desired. 
Random and ranked tie-breaking rules were considered in \citet{immorlica2006secretary} and \citet{karlin2015competitive}, respectively, in secretary settings.

Unlike the classical prophet scenario, which studies the optimization problem of a single decision maker, the setting of competing agents induces a game among multiple agents, were an agent's best strategy depends on the strategies chosen by others. Therefore, we study the equilibria of the induced games. In particular, we study the structure and quality of equilibrium in these settings and devise simple strategies that give agents high guarantees.

When the order of distributions is unknown in advance, calculating the optimal strategy is computationally hard. This motivates the use of simple and efficiently computed strategies that give good guarantees.

\subsection{Main Results and Techniques}
\label{sec:results}

For both random and ranked tie-breaking rules, we present simple single-threshold strategies for the agents that give them high guarantees. 
A single-threshold strategy specifies some threshold $T$, and selects any reward that exceeds $T$.

For $j=1, \ldots, n$, let $y_j$ be the $j$th highest reward. 

Under the random tie-breaking rule, we show a series of thresholds that have the following guarantee:

\begin{theorem*}(Theorem \ref{thm:prophet_threshold_inequality})
	For every $\ell=1, \ldots, n$, let $T^\ell = \frac{1}{k+\ell}\sum_{j=1}^{\ell} \E [y_j]$. 
	Then, for every agent, the single threshold strategy $T^{\ell}$ (i.e., select $v_t$ iff $v_t\geq T^\ell$) guarantees an expected utility of at least $T^{\ell}$.
\end{theorem*}

Two special cases of the last theorem are where $\ell=1$ and $\ell=k$. 
The case of $\ell=1$ implies that every agent can guarantee herself (in expectation) at least $\frac{1}{k+1}$ of the highest reward. 
The case of $\ell=k$ implies that every agent can guarantee herself (in expectation) at least $\frac{1}{2k}$ of the optimal social welfare (i.e., the sum of the highest $k$ rewards), which also implies that the social welfare in equilibrium is at least a half of the optimal social welfare. 

The above result is tight, as shown in Proposition~\ref{pro:lb_random}.

Similarly, for the ranked tie-breaking rule, we show a series of thresholds that have the following guarantee:
\begin{theorem*}(Theorem \ref{thm:prophet_serial_threshold_inequality})
	For every $i\leq n$ and $\ell=0, \ldots, n-i$, let $\hat{T}_i^\ell=\frac{1}{\ell+2}\sum_{j=i}^{i+\ell} \E [y_j]$. 
Then, for the $i$-ranked agent, the single threshold strategy $\hat{T}_i^\ell$ (i.e., select $v_t$ iff $v_t\geq \hat{T}_i^\ell$) guarantees an expected utility of at least $\hat{T}_i^\ell$.
\end{theorem*}

This result implies that for every $i$, the $i$-ranked agent can guarantee herself (in expectation) at least a half of the $i^{th}$ highest reward.
In Proposition~\ref{pro:lb_serial} we show that the last result is also tight.


Finally, we show that under the ranked tie-breaking rule, the equilibrium strategies of the (ordered) agents coincide with the decisions of a single decision maker who may select up to $k$ rewards in an online manner and wishes to maximize the sum of selected rewards.
Thus, the fact that every agent is aware of her position in the ranking allows them to coordinate around the socially optimal outcome despite the supposed competition between them.  
\begin{theorem*}(Corollary \ref{cor:welfare_prophet})
	Under the ranked tie-breaking rule, in every equilibrium of the $k$-agent game the expected social welfare is at least $1-O(\frac{1}{\sqrt{k}})$ of the optimal welfare.
\end{theorem*}

A similar phenomenon was observed in a related secretary setting, where the equilibrium strategy profile of a game with several ranked agents, induces an optimal strategy for a single decision maker who is allowed to choose several rewards and wishes to maximize the probability that the highest reward is selected \cite{matsui2016lower}. 

\subsection{Additional Related Literature}
The prophet problem and variants thereof has attracted a vast amount of literature in the last decade. For comprehensive surveys, see, e.g., the survey by 
\citet{hill1992survey} and the survey by 
\citet{lucier2017economic} which gives an economic view of the problem. 

A related well-known problem in the optimal stopping theory is the {\em secretary} problem, where the rewards are arbitrary but arrive in a random order. For the secretary problem a tight $1/e$-approximation has been established; for a survey, see, e.g., \cite{Ferguson89whosolved}.

Our work is inspired by a series of studies that consider scenarios where multiple agents compete over the rewards in secretary-like settings, where every agent aims to receive the highest reward.
\citet{karlin2015competitive} and \citet{immorlica2006secretary} considered the ranked- and the random tie-breaking rules, respectively, in secretary settings with competition.
For the ranked tie-breaking rule, \citet{karlin2015competitive} show that the equilibrium strategies take the form of {\em time-threshold strategies}; namely, the agent waits until a specific time $t$, thereafter competes over any reward that is the highest so far. The values of these time-thresholds are given by a recursive formula.
For the random tie-breaking rule, \citet{immorlica2006secretary} characterize the  Nash equilibria of the game and show that for several classes of strategies (such as threshold strategies and adaptive strategies), as the number of competing agents grows, the timing in which the earliest reward is chosen decreases. This confirms the argument that early offers in the job market are the result of competition between employers.

Competition among agents in secretary settings has been also studied by \citet{ezra2020competitive}, in a slightly different model. 
Specifically, in their setting, decisions need not be made immediately; rather, any previous reward can be selected as long as it is still available (i.e., has not been taken by a different agent). Thus, the competition is inherent in the model.

Another related work is the dueling framework by \citet{immorlica2011dueling}. 
One of their scenarios considers a 2-agent secretary setting, where one agent aims to maximize the probability of getting the highest reward (as in the classical secretary problem), and the other agent aims to outperform her opponent. 
They show an algorithm for the second agent that guarantees her a winning probability of at least $0.51$. 
They also establish an upper bound of $0.82$ on this probability. 

Other competitive models have been considered in the optimal stopping theory; see \cite{abdelaziz2007optimal} for a survey. 


The work of \citet{KleinbergW19} regarding matroid prophet problems is also related to our work.
They consider a setting where a single decision maker makes online selections under a matroid feasibility constraint, and show an  algorithm that achieve 1/2-approximation to the expected optimum for arbitrary matroids.
For the special case of uniform matroids, namely selecting up to $k$ rewards, earlier works of \citet{alaei2011bayesian} and \citet{HajiaghayiKS07} shows a approximation of $1-O(\frac{1}{\sqrt{k}})$ for the optimal solution. As mentioned above, the same guarantee is obtained in a setting with $k$ ranked competing agents.

\subsection{Structure of the Paper}
In Section \ref{sec:model} we define our model. 
In Sections \ref{sec:prophet_random}  and \ref{sec:prophet_lex} we present our results with respect to the
random tie-breaking rule, and the ranked tie-breaking rule, respectively.
We conclude the paper in Section \ref{sec:future} with future directions.

\section{Model	}
\label{sec:model}
We consider a prophet inequality variant, where a set of $n$ rewards, $v_1, \ldots, v_n$, are revealed online. While the values $v_1, \ldots, v_n$ are unknown from the outset, $v_t$ is drawn independently from a known probability distribution $F_t$, for $t\in[n]$, where $[n]=\{1,\ldots,n\}$. 
In the classical prophet setting, a single decision maker observes the realized reward $v_t$ at time $t$, and makes an immediate and irrevocable decision whether to take it or not. If she takes it, the game ends. Otherwise, the reward $v_t$ is lost forever, and the game continues with the next reward.

Unlike the classical prophet setting that involves a single decision maker, we consider a setting with $k$ decision makers (hereafter, agents) who compete over the rewards. 
Upon the revelation of reward $v_t$, every active agent (i.e., an agent who has not received a reward yet) may select it.
If a reward is selected by exactly one agent, then it is assigned to that agent. 
If the reward $v_t$ is selected by more than one agent, it is assigned to one of these agents either randomly (hereafter, {\em random tie-breaking}), or according to a predefined ranking (hereafter, {\em ranked tie-breaking}).
Agents who received rewards are no longer active.


A strategy of agent $i$, denoted by $S_i$, is a function that for every $t=1, \ldots, n$, decides whether or not to select $v_t$, based on $t$, the realization of $v_t$, and the set of active agents\footnote{One can easily verify that in our setting, additional information, such as the history of realizations of $v_1, \ldots, v_{t-1}$, and the history of selections and assignments, is irrelevant for future decision making.}. 
A strategy profile is denoted by $S=(S_1,\ldots,S_k)$.
We also denote a strategy profile by $S=(S_i,S_{-i})$, where $S_{-i}$ denotes the strategy profile of all agents except agent $i$.

Every strategy profile $S$ induces a distribution over assignments of rewards to agents. 
For ranked tie breaking, the distribution is with respect to the realizations of the rewards, and possibly the randomness in the agent strategies. For random tie breaking, the randomness is also with respect to the randomness in the tie-breaking. 
 
The utility of agent $i$ under strategy profile $S$, $u_i(S)$, is her expected reward under $S$; every agent acts to maximize her utility.

We say that a strategy $S_i$ guarantees agent $i$ a utility of $\alpha$ if $u_i(S_i,S_{-i}) \geq \alpha$ for every $S_{-i}$. 

\begin{definition}
A single threshold strategy $T$ is the strategy that upon the arrival of reward $v$, $v$ is selected if and only if the agent is still active and $v_t \geq T$.
\end{definition}

We also use the following equilibrium notions:

\begin{itemize}
	\item Nash Equilibrium (NE): A strategy profile $S=(S_1,\ldots, S_k)$ is a NE if for every agent $i$ and every strategy $S'_i$, it holds that $u_i(S'_i,S_{-i}) \leq u_i(S_i,S_{-i})$.
	
	\item Subgame perfect equilibrium (SPE): A strategy profile $S=(S_1,\ldots, S_k)$ is an SPE if $S$ is a NE for every subgame of the game. I.e. for every initial history $h$, $S$ is a NE in the game induced by history $h$. 
\end{itemize}

SPE is a refinement of NE; namely, every SPE is a NE, but not vice versa.

In the next sections, we let $y_j$ denote the random variable that equals the $j^{th}$ maximal reward among $\{v_1,\ldots,v_n\}$.

\section{Random Tie-Breaking}
\label{sec:prophet_random}
In this section we consider the random tie-breaking rule. 

We start by establishing a series of single threshold strategies that guarantee high utilities.
\begin{theorem}
		\label{thm:prophet_threshold_inequality}
		For every $\ell=1, \ldots, n$, let $T^\ell = \frac{1}{k+\ell}\sum_{j=1}^{\ell} \E [y_j]$. 
		Then, for every agent, the single threshold strategy $T^{\ell}$ (i.e., select $v_t$ iff $v_t\geq T^\ell$) guarantees an expected utility of at least $T^{\ell}$.
\end{theorem}	
\begin{proof}
	Fix an agent $i$. Let $S_{-i}$ be the strategies of all agents except agent $i$, and let $S=(T^\ell,S_{-i})$.
	Let $\assignn{i}{j}{S}$ denote the event that agent $i$ is assigned the reward $v_j$ in strategy profile $S$. I.e., $\assignn{i}{j}{S}$ is the event that agent $i$ competed over reward $v_j$ and received it according to the random tie-breaking rule. 
	For simplicity of presentation, we omit $S$ and write $\assign{i}{j}$.
	It holds that

\begin{eqnarray*}
	 u_i(S)  & = & \E\left[ \sum_{j=1}^{n}{v_j\cdot \Pr\left(\assign{i}{j}\right)}\right]\\
	& = & \E\left[\sum_{j=1}^{n}{(T^\ell+v_j-T^\ell) \Pr\left(v_j\geq T^\ell,\forall_{r<j}  \overline{\assign{i}{r}},\assign{i}{j}\right)}\right].
\end{eqnarray*}

Let $p=\sum_{j=1}^{n} \Pr(v_j\geq T^\ell,\forall_{j'<j}  \overline{\assign{i}{j'}},\assign{i}{j})$ (i.e., $p$ is the probability that agent $i$ receives some reward in strategy profile $S=(T^{\ell},S_{-i})$), and let $Z^{+}=\max\{Z,0\}$.
We can now write $u_i(S)$ as follows:
\begin{eqnarray*}
	u_i(S) 	& =  &p T^\ell + \E\left[\sum_{j=1}^{n}{(v_j-T^\ell)^+\Pr\left(\forall_{r<j}\overline{\assign{i}{r}},\assign{i}{j}\right)}\right]  \nonumber\\
	 &= & p\cdot T^\ell  + \E\Bigl[\sum_{j=1}^{n}(v_j-T^\ell)^+ \cdot \Pr\left(\forall_{r<j}\overline{\assign{i}{r}} \right)   \cdot\Pr\left(\assign{i}{j} \mid \forall_{r<j}\overline{\assign{i}{r}}\right)\Bigr] \nonumber\\
	 &\geq &	 p\cdot T^\ell + \E\Bigl[\sum_{j=1}^{n}(v_j-T^\ell)^+\cdot (1-p) \Bigr. \left. \cdot \Pr\left(\assign{i}{j} \mid \forall_{r<j}\overline{\assign{i}{r}}\right)\right] \nonumber\\
	 &\geq  &  p\cdot T^\ell +\frac{1-p}{k}\cdot \E\left[\sum_{j=1}^{n}(v_j-T^\ell)^{+}\right]. \nonumber 
\end{eqnarray*}		
The first inequality holds since the probability of not getting any reward until time $j$ is bounded by $1-p$ (i.e., the probability of not getting any reward). 
The last inequality holds since if $v_j-T^\ell\geq 0$ and agent $i$ is still active, the reward is selected, thus assigned with probability at least $1/k$. 
Since each term in the summation is non-negative, we get the following:
\begin{eqnarray*}		
	u_i(S)	& \geq & p\cdot T^\ell +\frac{1-p}{k}\cdot \E\left[\sum_{j=1}^{\ell}(y_j-T^\ell)^{+}\right] \nonumber \\ 
		& \geq  &p\cdot T^\ell + \frac{1-p}{k}\cdot \E\left[\sum_{j=1}^{\ell}y_j-\ell\cdot  T^\ell\right] \nonumber\\
		& =  &p\cdot T^\ell + \frac{1-p}{k}\cdot \left((k+\ell)\cdot T^\ell-\ell \cdot T^\ell\right)  = T^{\ell},\nonumber
\end{eqnarray*}
where the last equality follows by the definition of $T^\ell$.
\end{proof}

The special cases of $\ell=1$ and $\ell=k$ give the following corollaries:
\begin{corollary}
	The single-threshold strategy $T^k$ guarantees an expected utility of at least $\frac{1}{2k}\E[\sum_{i=1}^{k}y_i]$. 
\end{corollary}
\begin{corollary}
	The single-threshold strategy $T^1$ guarantees an expected utility of at least $\frac{1}{k+1}\E[y_1]$.
\end{corollary}

We now show that the bound in Theorem~\ref{thm:prophet_threshold_inequality} is tight.
\begin{proposition}\label{pro:lb_random}
	For every $\epsilon>0$ there exists an instance such that in the unique equilibrium of the game, no agent gets an expected utility of more than  $\frac{1}{k+\ell}\sum_{j=1}^{\ell} \E [y_j]+\epsilon$ for any $\ell \leq n$.
\end{proposition}
\begin{proof}
	Given an $\epsilon>0$, consider the following instance (depicted in Figure~\ref{fig:random}): 
	$$
	v_t=1 \mbox{ for all } t\leq n-1, ~\mbox{  and  }~  v_n=\begin{cases}
	\frac{k+\epsilon}{\epsilon} & \text{ w.p. } \epsilon\\
	0 & \text{ w.p. } 1-\epsilon
	\end{cases}
	$$
	One can easily verify that in the unique equilibrium $S$, all agents compete over the last reward, for an expected utility of 	 $1+\frac{\epsilon}{k}$. 
	It holds that for every agent $i$:
	$$u_i(S) =  1+\frac{\epsilon}{k} \leq 1+\epsilon =  
	\frac{\E[\sum_{j=1}^{\ell}y_j]}{k+\ell}+\epsilon.$$
	This example also shows that there are instances in which the social welfare in equilibrium is at most half the optimal welfare allocation. 
\end{proof}

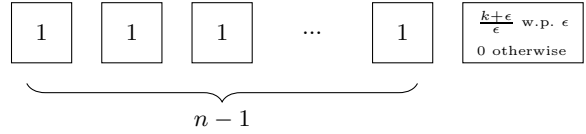
\begin{figure}[h!]
	\centering
	\begin{tikzpicture}[scale=0.4]
	\draw (0,0) rectangle node (C) {$1$} ++(2,2); 
	\draw (3,0) rectangle node (D) {$1$} ++(2,2); 
	\draw (6,0) rectangle node (F) {$1$} ++(2,2); 
	\draw [draw=white](9,0) rectangle node (dots) {...} ++(2,2); 
	\draw (12,0) rectangle node (A) {$1$} ++(2,2); 
	\draw (15,0) rectangle node (B)[align=left] {\tiny$\frac{k+\epsilon}{\epsilon}$ w.p. $\epsilon$\\\tiny$0$ otherwise} ++(4,2); 
	\draw[decorate,decoration={brace, amplitude=7pt, raise=13pt, mirror}]
	(C.south west) to node[black,below= 20pt] {$n-1$} (A.south east);%
	\end{tikzpicture}
	\caption{An example where the expected reward is no more than $\frac{1}{k+\ell}\sum_{j=1}^{\ell} \E [y_j]+\epsilon$}
	\label{fig:random}
\end{figure}

\section{Ranked Tie-Breaking}
\label{sec:prophet_lex}
In this section we consider the ranked tie-breaking rule, 
and present a series of single threshold strategies with their guarantees. 
We then show an interesting connection to the setting of a single agent that can choose up to $k$ rewards. 
We start by presenting the single threshold strategies.
 
\begin{theorem}\label{thm:prophet_serial_threshold_inequality}
	For every $i\leq n$ and $\ell=0, \ldots, n-i$, let $\hat{T}_i^\ell=\frac{1}{\ell+2}\sum_{j=i}^{i+\ell} \E [y_j]$. 
	The single threshold strategy $\hat{T}_i^\ell$ (i.e., select $v_t$ iff $v_t\geq \hat{T}_i^\ell$) guarantees an expected utility of at least $\hat{T}_i^\ell$ for the $i$-ranked agent.	
\end{theorem}
\begin{proof}
	Fix an agent $i$. Let $S_{-i}$ be the strategies of all agents except agent $i$, and let $S=(\hat{T}_i^\ell,S_{-i})$.
	Let $\assignn{i}{j}{S}$ denote the event that agent $i$ is assigned the reward $v_j$ in strategy profile $S$. I.e., $\assignn{i}{j}{S}$ is the event that agent $i$ competed over reward $v_j$ and received it according to the ranked tie-breaking rule. 
	For simplicity of presentation, we omit $S$ and write $\assign{i}{j}$.
	We bound the utility of agent $i$ under strategy profile $S$.
	\begin{eqnarray*}
	  u_i(S)  & = & \E\left[ \sum_{j=1}^{n}{v_j\cdot\Pr\left(\assign{i}{j}\right)}\right]    \nonumber \\ 
	 & =& \E\left[\sum_{j=1}^{n}{(\hat{T}_i^\ell+v_j-\hat{T}_i^\ell)\Pr\left(v_j\geq \hat{T}_i^\ell,\forall_{r<j} \overline{\assign{i}{r}},\assign{i}{j}\right)}\right]. \nonumber 
	\end{eqnarray*}

Let $p=\sum_{j=1}^{n} \Pr(v_j\geq \hat{T}_i^\ell,\forall_{r<j}  \overline{\assign{i}{r}},\assign{i}{j})$ (i.e., $p$ is the probability that agent $i$ receives some reward in strategy profile $S=(\hat{T}_i^\ell,S_{-i})$), and let $Z^{+}=\max\{Z,0\}$.
We can now write $u_i(S)$ as follows:
\begin{eqnarray}
	 u_i(S)	  & = & p \cdot\hat{T}_i^\ell + \E\left[\sum_{j=1}^{n}{(v_j-\hat{T}_i^\ell)^+\Pr\left(\forall_{r<j} \overline{\assign{i}{r}},\assign{i}{j}\right)}\right] \nonumber \\
	& \geq & p \cdot  \hat{T}_i^\ell \label{eq:1} + \E\Bigl[\sum_{j=1}^{n}(v_j-\hat{T}_i^\ell)^+ \Bigr.  
	 \left. \cdot (1-p)\cdot\Pr\left(\assign{i}{j} \mid \forall_{r<j}\overline{\assign{i}{r}}\right)\right] \nonumber\\
	& \geq & p\cdot \hat{T}_i^\ell + (1-p)\cdot \E\left[\sum_{j=i}^{n}(y_j-\hat{T}_i^\ell)^+\right] \label{eq:2}\\
	& \geq & p\cdot \hat{T}_i^\ell + (1-p)\cdot \E\left[\sum_{j=i}^{i+\ell}(y_j-\hat{T}_i^\ell)\right] \nonumber\\
	& = & p\cdot \hat{T}_i^\ell + (1-p)\cdot \left( \E\left[\sum_{j=i}^{i+\ell}y_j\right]-(\ell+1)\hat{T}_i^\ell \right)\nonumber\\
	& = & p\cdot \hat{T}_i^\ell + (1-p)\cdot \left((\ell+2)\cdot \hat{T}_i^\ell-(\ell+1)\cdot\hat{T}_i^\ell \right) = \hat{T}_i^\ell. \nonumber
	\end{eqnarray}
	
	
	Inequality \eqref{eq:1} holds since the probability of not getting any reward until time $j$ is bounded by $1-p$ (i.e., the probability of not getting any reward). 
	Inequality \eqref{eq:2} holds since there are at most $i-1$ agents that are ranked higher than agent $i$, therefore there are at most $i-1$ rewards that can be selected but not assigned to agent $i$.
	Finally, the last equality holds by the definition of $\hat{T}_i^\ell$.		
\end{proof}

The special case of Theorem~\ref{thm:prophet_serial_threshold_inequality} where $\ell=0$ gives the following corollary.
\begin{corollary}
For every $i$, the threshold strategy $\hat{T}_i^0$ guarantees an expected utility of $\frac{\E[y_i]}{2}$ for the $i$-ranked agent.
\end{corollary}

We next show that the bound in Theorem~\ref{thm:prophet_serial_threshold_inequality} is tight.
\begin{proposition}\label{pro:lb_serial}
	For every $\epsilon>0$ and every $i \leq n$, there exists an instance such that in the unique equilibrium of the game, the $i$-ranked agent gets an expected utility of at most $\frac{1}{\ell+2}\sum_{j=i}^{i+\ell} \E [y_j]+\epsilon$ for every $\ell \leq n-i$.
\end{proposition}

\begin{proof}
	Given some $\epsilon>0$ and $i \leq n$, consider the following instance (depicted in Figure~\ref{fig:ranked}): 
	$$
	v_t = 
	\begin{cases}
	\infty & \text{ for } t < i\\
	1 & \text{ for } i \leq t < n \\
	\frac{1+\epsilon}{\epsilon} \mbox{ w.p. } \epsilon, \mbox{ and } 0  \mbox{ w.p. } 1-\epsilon& \text{ for } t = n
	\end{cases}
	$$
	One can easily verify that in the unique equilibrium of the game, agents $1, \ldots, i-1$ will be assigned rewards $v_1,\ldots,v_{i-1}$, and agent $i$ will be assigned the last reward $v_n$ for an expected utility of $1+\epsilon$.
	It holds that:
	$$ u_i(S) = 1+\epsilon = \frac{\E[\sum_{j=i}^{i+\ell}y_j]}{2+\ell}+\epsilon.$$
\end{proof}

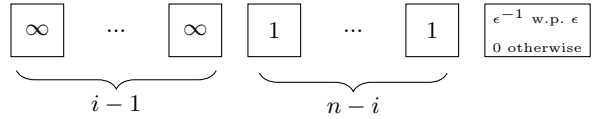
\begin{figure}[h!]
	\centering
	\begin{tikzpicture}[scale=0.35]
	\draw (0,0) rectangle node (G) {$\infty$} ++(2,2); 
	\draw [draw=white](3,0) rectangle node (dots) {...} ++(2,2); 
	\draw (6,0) rectangle node (J) {$\infty$} ++(2,2); 
	\draw (9,0) rectangle node (C) {$1$} ++(2,2); 
	\draw [draw=white](12,0) rectangle node (dots) {...} ++(2,2); 
	\draw (15,0) rectangle node (A) {$1$} ++(2,2); 
	\draw (18,0) rectangle node (B)[align=left] {\tiny${\epsilon}^{-1}$ w.p. $\epsilon$\\\tiny$0$ otherwise} ++(4,2); 
	\draw[decorate,decoration={brace, amplitude=7pt, raise=10pt, mirror}]
	(C.south west) to node[black,below= 16pt] {$n-i$} (A.south east);%
	\draw[decorate,decoration={brace, amplitude=7pt, raise=10pt, mirror}]
	(G.south west) to node[black,below= 16pt] {$i-1$} (J.south east);%
	\end{tikzpicture}
	\caption{An example where the expected reward for agent $i$ is no more than $\frac{1}{\ell+2}\sum_{j=i}^{i+\ell} \E [y_j]+\epsilon$}
	\label{fig:ranked}
\end{figure}

We next show that for any instance, the set of rewards assigned to the $k$ competing agents in equilibrium coincides with the set of rewards that are chosen by the optimal algorithm for a single decision maker who can choose up to $k$ rewards and wishes to maximize their sum. 
\citet{KleinbergW19} show that the only optimal strategy of such a decision maker, takes the form of $nk$ dynamic thresholds, $\{T_t^i\}_{i,t}$ for all $t\leq n$ and $i \leq k$, so that the agent accepts reward $v_t$ if $v_t \geq T_t^i$, where $k-i$ is the number of rewards already chosen (i.e., $i$ is the number of rewards left to choose)\footnote{The uniqueness holds for distributions with no mass points. For distributions with mass points, whenever $v_t=T_t^i$, the decision maker is indifferent between selecting and passing.}. Moreover, they show that these thresholds are monotone with respect to $i$.


With the characterization of the strategy of a single decision maker who can choose up to $k$ rewards, we can characterize the unique SPE for the $k$-agent game\footnote{The SPE is unique up to cases where $T_j^i=v_t$; in these cases the agent is indifferent.}.
\begin{theorem}\label{thm:propht many lex is like a single agent}
Let $\{T_t^i\}_{i \in [k],t \in [n]}$ be the optimal strategy of a single decision maker who may choose up to $k$ rewards and wishes to maximize their sum. The unique SPE of the $k$-agent game is for agent $i$ to accept $v_t$ iff $v_t \geq T_t^{i'+1}$, where $i'\leq i$ is the rank of agent $i$ among the active agents. This SPE is unique up to cases where $v_t = T_t^{i'}$.
\end{theorem}
\begin{proof}
	Let $S^i$ denote the optimal strategy of the single agent who may choose up to $i$ rewards, as described above. 
	Let $S_i$ be the strategy of agent $i$ as described in the assertion of the theorem.
	We prove by induction that for every $i \in [k]$, the rewards that are chosen by agents $1,\ldots,i$ correspond to the rewards chosen by a single decision maker, who may choose up to $i$ rewards, and uses strategy $S^i$. 
	For the case of $i=1$, the claim holds trivially. 
	Assume the claim holds for any number of agents smaller than $i$. 
	Since agent $i$ has no influence on the rewards received by agents $1,\ldots,i-1$, we may assume that agents $1,\ldots,i-1$ are playing according to strategies $S_1,\ldots,S_{i-1}$.
	
	For every $i\in [k]$, the total utility of agents $1,\ldots,i$ is bounded by the utility of the single decision maker $u(S^i)$, since the single decision maker can simulate a game with $i$ competing agents. Hence, by the induction hypothesis, agent $i$ can obtain a utility of at most $u(S^i)-u(S^{i-1})$. By playing according to $S_i$, we are guaranteed that whenever at least $j$ agents are still active, any reward $v_t$ such that $v_t \geq T_t^{j}$ will be taken by one of the agents. Thus, when every agent $i$ is playing according to $S_i$, players $1, \ldots, i$ play according to $S^i$. Consequently, their total utility is $u(S^i)$, and the utility of agent $i$ is then maximal. The uniqueness (up to the cases where $v_j = T_j^{i'}$) is by the uniqueness of the optimal strategy of the single decision maker.     
\end{proof}

We note that by Theorem~\ref{thm:prophet_serial_threshold_inequality} it holds that
in the unique SPE described in Theorem~\ref{thm:propht many lex is like a single agent}, every agent  $i$ receives at least $\max_{\ell=0}^{n-i} \frac{1}{\ell+2}\sum_{j=i}^{i+\ell} \E [y_j]$. 

Using the results of \citet{alaei2011bayesian} regarding a single decision maker choosing $k$ rewards, we deduce an approximation of the social welfare in equilibrium:
\begin{corollary} \label{cor:welfare_prophet}
	In SPE of the $k$ agent prophet game, the expected social welfare is at least $1-O(\frac{1}{\sqrt{k}})$ of the optimal welfare.
\end{corollary}

\section{Discussion and Future Directions}
\label{sec:future}
In this work, we study the effect of competition in prophet settings.
We show that under both random and ranked tie-breaking rules, agents have simple strategies that grant them high guarantees, ones that are tight even with respect to equilibrium profiles under some distributions.

Under the ranked tie-breaking rule, we show an interesting correspondence between the equilibrium strategies of the $k$ competing agents and the optimal strategy of a single decision maker that can select up to $k$ rewards. 
It would be interesting to study whether this phenomenon applies more generally, and what are the conditions under which it holds.

Below we list some future directions that we find particularly natural. 
\begin{itemize}
	\item Study competition in additional problems related to optimal stopping theory, such as Pandora's box \cite{weitzman1979optimal}.
	\item Study competition in prophet (and secretary) settings under additional tie-breaking rules, such as random tie breaking with non-uniform distribution, and tie-breaking rules that allow to split rewards among agents.
	\item Study competition in  scenarios where agents can choose multiple rewards, under some feasibility constraints (such as matroid or downward-closed feasibility constraints). 
	\item Consider prophet settings with the objective of outperforming the other agents, as in \cite{immorlica2011dueling}, or different agents' objectives.
	\item Consider competition settings with non-immediate decision making, as in \cite{ezra2020competitive}.
\end{itemize}

\section*{Acknowledgement}
The work was partially supported by the European Research Council (ERC) under the European Union's Horizon 2020 research and innovation program (grant agreement No. 866132, 740282), and by the Israel Science Foundation (grant number 317/17).
%
%
%
 \bibliographystyle{plainnat}
 \bibliography{prophet}
\end{document}